\documentclass[11pt]{amsart}

\usepackage{a4wide}
\usepackage{amsmath,amssymb}
\usepackage{mathabx}
\usepackage{bbm}
\usepackage{color}

\usepackage{scalerel,stackengine}
\stackMath
\newcommand\reallywidehat[1]{%
\savestack{\tmpbox}{\stretchto{%
  \scaleto{%
    \scalerel*[\widthof{\ensuremath{#1}}]{\kern-.6pt\bigwedge\kern-.6pt}%
    {\rule[-\textheight/2]{1ex}{\textheight}}%WIDTH-LIMITED BIG WEDGE
  }{\textheight}%
}{0.5ex}}%
\stackon[1pt]{#1}{\tmpbox}%
}

\theoremstyle{plain}
\newtheorem{theorem}{Theorem}
\newtheorem{prop}[theorem]{Proposition}
\newtheorem{lemma}[theorem]{Lemma}
\newtheorem{coro}[theorem]{Corollary}

\theoremstyle{definition}
\newtheorem{definition}[theorem]{Definition}
\newtheorem{remark}[theorem]{Remark}
\newtheorem{example}[theorem]{Example}

\newcommand{\Z}{{\mathbb Z}}

\newcommand{\R}{{\mathbb R}}
\newcommand{\N}{{\mathbb N}}
\newcommand{\C}{{\mathbb C}}
\newcommand{\im}{{\mathrm{i}}}
\newcommand{\e}{\operatorname{e}}
\newcommand{\mc}{\mathcal}

\newcommand{\eps}{\varepsilon}
\newcommand{\cM}{{\mathcal M}}
\newcommand{\cR}{{\mathcal R}}
\newcommand{\WAP}{\mathcal{W}\hspace*{-1pt}\mathcal{AP}}
\newcommand{\SAP}{\mathcal{S}\hspace*{-2pt}\mathcal{AP}}
\newcommand{\exend}{\hfill $\Diamond$}

\begin{document}

\title{A note on measures
  vanishing at infinity}

\author{Timo Spindeler}
\address{Fakult\"at f\"ur Mathematik, Universit\"at Bielefeld, \newline
\hspace*{\parindent}Postfach 100131, 33501 Bielefeld, Germany}
\email{tspindel@math.uni-bielefeld.de}

\author{Nicolae Strungaru}
\address{Department of Mathematical Sciences, MacEwan University \\
10700 -- 104 Avenue, Edmonton, AB, T5J 4S2\\
and \\
Department of Mathematics\\
Trent University \\
Peterborough, ON
and \\
Institute of Mathematics ``Simon Stoilow''\\
Bucharest, Romania}
\email{strungarun@macewan.ca}
\urladdr{http://academic.macewan.ca/strungarun/}

\begin{abstract}
In this paper, we review the basic properties of measures vanishing at infinity and prove a version of the Riemann--Lebesgue lemma for Fourier transformable measures.
\end{abstract}

\keywords{Fourier--Stieltjes transform, unbounded measures, Riemann--Lebesgue lemma}

\subjclass[2010]{43A25, 52C23}

\maketitle

\section{Introduction}

Physical diffraction plays a central role in the study of the atomic structure of solids. The discovery of quasicrystals by Shechtman \cite{She} in 1984 led to an increased interest into structures which fail to be periodic, but still exhibit long range order, typically shown by a large Bragg spectrum.

Mathematically, the diffraction is described as follows: given a point set $\Lambda$, representing the positions of atoms in an ideal crystal, or more generally a measure $\mu$, we construct a new measure $\gamma$ called the \emph{autocorrelation} of $\Lambda$ or $\mu$ respectively. This measure is  positive definite and hence Fourier transformable \cite{ARMA1,BF,ARMA}.
Its Fourier transform $\widehat{\gamma}$ is a positive measure which models the diffraction of our structure \cite{HOF,TAO}.

As any measure $\widehat{\gamma}$ has a Lebesgue decomposition
\[
\widehat{\gamma}=\widehat{\gamma}_{\text{pp}}+\widehat{\gamma}_{\text{ac}}+\widehat{\gamma}_{\text{sc}} \,,
\]
with respect to the Haar measure $\theta_{\widehat{G}}$ into the pure point, absolutely continuous and singular continuous components.

In many particular cases, we know quite a bit about the pure point and continuous spectrum. One way to get information about these components is via the \emph{Eberlein decomposition}: the measure $\gamma$ can be decomposed \cite{ARMA,MoSt,NS4} into two positive definite measures $\gamma_{\mathsf{s}}$ and $\gamma_{0}$, called the {\it strong} and {\it null weakly almost periodic} parts of $\gamma$, such that
\[
\reallywidehat{\gamma_{\mathsf{s}}} =\left( \widehat{\gamma} \right)_{\text{pp}} \quad \mbox{ and } \quad \reallywidehat{\gamma_{0}} =\left( \widehat{\gamma} \right)_{\text{c}} \,.
\]
By using this approach, progress has been made towards understanding the pure point and continuous spectra of measures with lattice support \cite{BA,TAO}, and with Meyer set support \cite{LS2,NS1,NS2,NS5,NS11} and few examples of measures with FLC support \cite{LS2,NS4}. More recently, this decomposition has been studied via spectral decomposition of dynamical systems \cite{JBA}.

These methods also led to a good understanding of models with pure point diffraction. Pure point diffraction was characterised in terms of almost periodicity of the autocorrelation measure  \cite{BM,Gouere-1,ARMA,MoSt,NS3}, in terms of almost periodicity of the underlying structure \cite{MS} or in terms of the pure point dynamical spectrum \cite{BL,LMS1}. Various systems with pure point diffraction, such as regular model sets and weighted Dirac combs with Meyer set support \cite{BM,HOF,LR,CR,CRS,CRS2,NS11,Martin2}, weak model sets of maximal density \cite{BHS,KR}, stationary processes \cite{DM,LM} and various deformations \cite{BL2,LS}, or almost periodic measures \cite{LS2} have been studied. For more general overviews of these results we recommend \cite{BG1,TAO,DL3,BL3,BG2}.

In contrast, little is known in general about the continuous spectra and systems with purely continuous spectrum. The decomposition
\[
\widehat{\gamma}_{\text{c}}= \widehat{\gamma}_{\text{ac}} + \widehat{\gamma}_{\text{sc}} \,,
\]
is only understood in few particular examples, and almost nothing is known in general about systems which don't have some underlying lattice structure. For example, the pinwheel tiling is known to have only one trivial Bragg peak \cite{MSP}, therefore the diffraction is essentially continuous, but nothing else is known about the structure of this diffraction, except its circular symmetry.

One interesting example for which the continuous spectrum is well understood is the Thue--Morse model, see \cite{BG2,Que} for details.

The study of the Thue--Morse example is simplified by the fact that this model leads to a class of weighted combs with lattice support. Moreover, by choosing the right weights, we can get a measure $\omega$ whose diffraction is exactly the continuous diffraction spectra of the Thue--Morse. This also implies that the autocorrelation $\gamma$ of $\omega$ is a null-weakly almost periodic measure with lattice support. A relatively simple computation shows that the autocorrelation coefficients of $\gamma$ are not vanishing at infinity \cite{BG1,BG2}. Since $\gamma$ is supported on a lattice, the Riemann--Lebesgue lemma can be applied to conclude that the diffraction cannot be absolutely continuous, and a relatively simple argument can be then used to conclude that $\widehat{\gamma}_{\text{ac}}=0$ \cite{BG1,BG2}.

The Thue--Morse example can be generalised to a large class of similar models \cite{Kea}. By a similar argument, all these models have been shown to have purely singularly continuous diffraction for a suitable choice of the weights \cite{BG3,BGG,TAO}. Other similar models with mixed singular spectrum have been studied \cite{BG,BGa}.

In all these examples, the key is the fact that, for a measure supported on a lattice with coefficients not vanishing at infinity, the Riemann--Lebesgue lemma can be applied to conclude that the Fourier transform cannot be absolutely continuous.

It is worth emphasising that this approach does not work for Fourier transformable measures which are not supported on lattices. It is our goal to give a version of the Riemann--Lebesgue lemma for arbitrary Fourier transformable measures, and to show that for measures with lattice support this version is equivalent to the standard one.

The paper is organised as it follows: In Section \ref{sec:measures van at infinity}, we recall the notion of measures vanishing at infinity and study the basic properties of this class. Next, in Section \ref{sect:On the}, we provide the main theorem of the paper, the Riemann--Lebesgue lemma for measures together with few consequences. We also recall a classical result about absolutely continuous positive definite measures, which complements our main result. Further examples are provided in Section \ref{sect:exa}. We complete the paper by taking a close look, in Section \ref{sect:raj}, at the Fourier duality between Fourier transformable measures vanishing at infinity and Rajchman measures.

\newpage

\section{Preliminaries}

For the entire paper, $G$ denotes a $\sigma$-compact LCA group. We will denote by $C_U(G)$ the space of uniformly continuous and bounded functions on $G$. $C_0(G)$ will represent the subspace of $C_U(G)$ consisting of functions vanishing at infinity. Let us recall that $f$ is said to be vanishing at infinity if, for each $\eps>0$, there exists a compact set $K_{\eps}$ such that
\begin{displaymath}
\left| f(x) \right| < \eps \quad \text{for all $x \notin K_{\eps}$}.
\end{displaymath}
$C_U(G)$ is a Banach space with respect to the $\| . \|_\infty$-norm, and $C_0(G)$ is a closed subspace of $C_U(G)$.

\medskip

As usual, $C_c(G)$ denotes the subspace of $C_U(G)$ of functions with compact support. This is not closed with respect to $\| . \|_\infty$, but it is closed with respect to the \emph{inductive topology}, which is defined as follows: For each compact set $K \subseteq G$, the space
\begin{displaymath}
C(G:K):= \{ f \in C_c(G)\, |\, \text{supp}(f)\subseteq K \}
\end{displaymath}
is a Banach subspace of $C_U(G)$. As we have
\begin{displaymath}
C_c(G) = \bigcup_{\substack{K \subseteq G \\ K \mbox{ compact }}} C(G:K) \, ,
\end{displaymath}
$C_c(G)$ can be seen as the inductive limit of the spaces $C(G:K)$. The topology defined by this inductive limit is called the inductive topology.

\medskip

By the Riesz--Markov theorem (see \cite[Appendix]{CRS2} for details), the space of regular Radon measures on $G$ can be identified with the dual space of $C_c(G)$, when this is equipped with the inductive limit topology. Because of this, we will think of measures as continuous linear functionals on $C_c(G)$, and denote the duality between a measure $\mu$ and a function $f \in C_c(G)$ by
\begin{displaymath}
\left\langle f, \mu\right\rangle\ := \int_G f d \mu \,.
\end{displaymath}
The convolution between two functions $f,g \in C_c(G)$ is defined via
\begin{displaymath}
(f*g)(x) =\int_G f(x-t) g(t)\, dt =\int_G f(t) g(x-t)\, dt.
\end{displaymath}
Similarly, for a measure $\mu$ and a function $f \in C_c(G)$, we define the convolution $\mu*f$ by
\begin{displaymath}
(\mu*f)(x)= \int_G f(x-t)\, d \mu(t) .
\end{displaymath}

\medskip

Let us next recall the definition of a translation bounded measure.

\begin{definition} A measure $\mu$ is called \emph{translation bounded} if for all $f \in C_c(G)$ we have $\mu*f \in C_U(G)$.

We denote the space of translation bounded measures by $\cM^\infty(G)$.
\end{definition}

\begin{remark} By \cite{ARMA1}, a measure $\mu$ is translation bounded if and only if, for all compact sets $K \subseteq G$, we have
\begin{displaymath}
\|  \mu \|_K := \sup_{x \in G} \left| \mu \right| (x+K) < \infty \,,
\end{displaymath}
where $\left| \mu \right|$ denotes the variation measure for $\mu$ (see \cite[Sec. 6.5]{Ped} for details).

In \cite{BM}, the authors showed that it suffices to check that $\| \mu \|_K < \infty$ for a single compact set with non-empty interior.
\exend
\end{remark}

\medskip
Next, we give a brief review of Fourier transformability for measures. For a more detailed overview, we refer the reader to \cite{MoSt}.

As usual, $K_2(G)$ is defined as the subspace of $C_c(G)$ spanned by the set $\{f*g\ |\ f,g\in C_c(G)\}$.
Note that, by the polarisation identity, we have
\begin{displaymath}
K_2(G) = \mbox{Span} \{ f*\widetilde{f}\, |\, f \in C_c(G) \} \,.
\end{displaymath}

\begin{definition} A measure $\mu$ is called \emph{Fourier transformable} if there exists a measure $\widehat{\mu}$ on $\widehat{G}$ such that, for all $f \in C_c(G)$ we have $\widecheck{f} \in L^2( \widehat{\mu})$ and
\begin{displaymath}
\left\langle f*\widetilde{f}, \mu\right\rangle \ =\ \left\langle \big| \widecheck{f}\, \big|^2 , \widehat{\mu} \right\rangle \,.
\end{displaymath}
\end{definition}

Note that the definition is equivalent to $\widecheck{g} \in L^1( \widehat{\mu})$ for all $g \in K_2(G)$ together with
\begin{displaymath}
\left\langle g, \mu \right\rangle\ =\ \left\langle  \widecheck{g} , \widehat{\mu} \right\rangle \,.
\end{displaymath}

\medskip

We complete this section by briefly reviewing the notions of strongly, weakly and null weakly-almost periodic measures and their relation to the Fourier transform of measures. For a comprehensive review of this, we refer the reader to \cite{MoSt}.

As the definitions for measures are obtained by convolutions from the equivalent notions for functions, we start by introducing these concepts for functions.

\begin{definition} A function $f \in C_U(G)$ is called \emph{strongly almost periodic} or \emph{Bohr almost periodic} if the hull $\{ T_t f\, |\, t \in G \}$
is pre-compact in $(C_U(G), \| . \|_\infty)$, where the translation $T_tf$ is defined via
\begin{displaymath}
T_tf(x):=f(x-t) \,.
\end{displaymath}

Further, $f \in C_U(G)$ is called \emph{weakly almost periodic} if the hull $\{ T_t f\, |\, t \in G \}$
is pre-compact in the weak topology of $(C_U(G), \| . \|_\infty)$.

We denote by $S\hspace*{-1pt}AP(G)$ and $W\hspace*{-2pt}AP(G)$ the spaces of strongly respectively weakly almost periodic functions.
\end{definition}

It is obvious that $S\hspace*{-1pt}AP(G) \subseteq W\hspace*{-2pt}AP(G)$.

To define null weakly almost periodicity, we first need  to define the concept of a mean.

\begin{prop}\cite{Eb,MoSt} Let $f \in W\hspace*{-2pt}AP(G)$ and $(A_n)_{n\in\N}$ be a van Hove sequence. Then, the following limit exists:
\begin{displaymath}
 \lim_{n\to\infty} \frac{1}{|A_n|} \int_{A_n} f(x)\, dx \, =: \, M(f) \,.
\end{displaymath}
\end{prop}

\begin{definition} For $f \in W\hspace*{-2pt}AP(G)$, the number $M(f)$ is called the \emph{mean of $f$}.

A function $f$ is called \emph{null weakly almost periodic} if $f \in W\hspace*{-2pt}AP(G)$ and $M(|f|)=0$. We denote the space of null weakly almost periodic functions by $W\hspace*{-2pt}AP_0(G)$.
\end{definition}

Now, we can carry these definitions to measures.

\begin{definition} A measure $\mu \in \cM^\infty(G)$ is called \emph{strongly, weakly} respectively \emph{null weakly almost periodic} if, for all $f \in C_c(G)$, the function $\mu*f$ is strongly, weakly respectively null weakly almost periodic.

We denote by $\SAP(G),\WAP(G)$ and $\WAP_0(G)$ the spaces of strongly, weakly respectively null weakly almost periodic measures.
\end{definition}

\medskip

The importance of the spaces of almost periodic measures to diffraction theory is given by the next two theorems.

\begin{theorem} \cite{ARMA,MoSt}
\begin{displaymath}
\WAP(G) = \SAP(G) \bigoplus \WAP_0(G) \,.
\end{displaymath}
In particular, any measure $\mu \in \WAP(G)$ can be written uniquely as
\begin{equation} \label{eq:h}
\mu=\mu_s+\mu_0 \quad \text{ with } \mu_s \in \SAP(G) \text{ and } \mu_0 \in \WAP_0(G) \,.
\end{equation}
\end{theorem}

We will refer to the decomposition $\mu=\mu_s+\mu_0$ as the \emph{Eberlein decomposition} for weakly almost periodic measures.

\begin{theorem} \cite{MoSt}
Let $\mu \in \cM^\infty(G)$ be Fourier transformable. Then $\mu \in \WAP(G)$, $\mu_s, \mu_0$ are Fourier transformable, and
\begin{displaymath}
\reallywidehat{\mu_s}= (\widehat{\mu})_{pp} \, ; \quad \reallywidehat{\mu_0}= (\widehat{\mu})_{c} \,.
\end{displaymath}
\end{theorem}

\section{Measures vanishing at infinity}\label{sec:measures van at infinity}

In this section, we define and study the basic properties of measures vanishing at infinity. Later, we will need these measures for the formulation of the Riemann--Lebesgue lemma for measures. Some of the results of this section are stated (without proof) in \cite[pp. 5-6]{BF}.

\begin{definition} \cite[Def.~1.14]{BF}
A measure $\mu\in\mc M(G)$ is \emph{vanishing at infinity} if
\begin{displaymath}
\mu*f \in C_0(G) \quad \text{ for all } f\in C_c(G).
\end{displaymath}
In this case, we write $\mu\in \mc M_0^{\infty}(G)$.
\end{definition}

\smallskip

The following observation is an exercise in \cite{BF} and is trivial to prove. 

\begin{remark}\label{rem:BF exe 1.17}\cite[Exercise~1.17]{BF} Let $\mu$ be a positive measure. Then $\mu \in  \mc M_0^{\infty}(G)$ if and only if for all compact sets $K \subset G$ the function $x \to \mu(x+K)$ is vanishing at infinity.
\end{remark}

Exactly as for translation boundedness, it is immediate to see that it suffices to check the condition from Remark~\ref{rem:BF exe 1.17} for one compact sets $K$ with non-empty interior. 

\medskip

We next look at a basic example of a measure vanishing at infinity, which can be found in the literature in various places \cite{TAO,NS1,NS2,NS5}:

\begin{example}\label{ex:a} Let
\begin{displaymath}
 \mu = \displaystyle\sum_{\substack{n \in \Z \\ n\neq 0}} \left( \delta_{n+\frac1n}-\delta_n \right)  \,.
\end{displaymath}
Then $\mu \in \cM_0^\infty(G)$. For a proof we refer to \cite[Ex. 5.8]{NS1}.
\end{example}

\medskip

When we deal with the Fourier transform of measures, we need to deal with test functions in $K_2(G)$ instead of $C_c(G)$. We start by proving that the definition is the same if one uses $K_2(G)$ as the space of test functions.

\begin{lemma} \label{lem:a}
Let $\mu\in\mc M^{\infty}(G)$. Then, $\mu\in\mc M_0^{\infty}(G)$
if and only if
\begin{equation} \label{eq:c}
\mu*f \in C_0(G) \quad \text{ for all } f\in K_2(G).
\end{equation}
\end{lemma}
\begin{proof}
It is trivial that~(\ref{eq:c}) is necessary because $K_2(G)
\subseteq C_0(G)$. To prove that it is also sufficient, let $f\in C_c(G)$. By assumption, for all $g \in C_c(G)$ we have
\begin{displaymath}
\mu*(f*g)  \in C_0(G)  \,.
\end{displaymath}
Let $(h_{\alpha})_{\alpha}$ be an approximate identity for $(C_U(G), *)$, consisting of
functions from $C_c(G)$. Then,
\begin{equation} \label{eq:a}
(\mu*f)*h_{\alpha} = \mu*(f*h_{\alpha}) \in C_0(G) \,.
\end{equation}
The translation boundedness of $\mu$ implies $\mu*f\in C_U(G)$, therefore
\begin{equation}  \label{eq:b}
\lim_{\alpha}\ (\mu*f)*h_{\alpha}\ = \ \mu*f \quad
\text{uniformly}.
\end{equation}
Now,~(\ref{eq:a}) and~(\ref{eq:b}) together with the completeness of $C_0(G)$ imply $\mu*f\in C_0(G)$.
\end{proof}

\medskip

The rest of this section is devoted to some basic properties of measures vanishing at infinity.
As to that, the next two lemmas state useful examples of sets that contain $\cM_0^{\infty}(G)$.
First, we state a simple, but important consequence of vanishing at infinity (compare \cite[Remark 1.15]{BF}).

\begin{lemma}
Every measure $\mu$ that is vanishing at infinity is
translation bounded.
\end{lemma}
\begin{proof}
The claim follows immediately from $C_0(G)\subseteq C_U(G)$.
\end{proof}

\medskip
Next, we show that for non-compact groups the measures vanishing at infinity are automatically null weakly almost periodic.

\begin{lemma}
Let $G$ be non-compact. Then, $\mc M_0^{\infty}(G) \subseteq \WAP_0(G)$.
\end{lemma}
\begin{proof}
Due to \cite[Thm 11.1]{Eb} we know that
\begin{displaymath}
C_0(G) \subseteq W\hspace*{-2pt}AP(G).
\end{displaymath}
Hence, we are left to show that
\begin{displaymath}
M(|f|)=0 \quad \text{ for all } f\in C_0(G).
\end{displaymath}
Let $(A_n)_{n\in\N}$ be a van Hove sequence in $G$, see \cite{Martin2}. As $f\in C_0(G)$, for all
$\eps>0$ there is a compact set $K_{\eps}$ such that
\begin{displaymath}
|f(x)|<\eps \quad \text{ for all } x\notin K_{\eps}.
\end{displaymath}
Thus,
\begin{displaymath}
\begin{split}
M(|f|)
  &= \lim_{n\to\infty} |A_n|^{-1} \int_{A_n} |f|\, d\theta                                                          \\
  &\leq \lim_{n\to\infty} |A_n|^{-1} \int_{K_{\eps}} |f|\, d\theta\ +\
    \lim_{n\to\infty} |A_n|^{-1} \int_{A_n\setminus K_{\eps}} |f|\, d\theta                                  \\
  &\leq 0 + \eps \cdot \lim_{n\to\infty} \frac{|A_n\setminus K_{\eps}|}{|A_n|}
     \leq \eps,
\end{split}
\end{displaymath}
since $\lim_{n\to\infty} |A_n|^{-1} =0$ follows from $G$ non-compact.
\end{proof}

\medskip

On the other hand, we can find sufficient conditions for a measure to be vanishing at infinity, which are often easy to check. For instance, every finite measure is vanishing at infinity. This produces a large class of examples, which includes $L^p(G)$ for $p \geq 1$.

\begin{prop} \label{prop:b}
Every finite measure $\mu\in \mc M(G)$ is vanishing at infinity.
\end{prop}
\begin{proof}
Let $f\in C_c(G)$ and $\eps>0$. Since $\mu$ is regular, there is a compact
set $\widetilde{K_{\eps}}\subseteq G$ such that
\begin{displaymath}
\mu( G\setminus \widetilde{K_{\eps}}) < \frac{\eps}{1+\|f\|_{\infty}}.
\end{displaymath}
Set $K_f:=\operatorname{supp}(f)$. Then,
\begin{displaymath}
\begin{split}
|\mu*f(t)|
  &= \left| \int_G f(t-s)\, d\mu(s) \right|                                                                                   \\
  &= \left| \int_{ \widetilde{K_{\eps}}} f(t-s)\, d\mu(s)
     + \int_{G\setminus \widetilde{K_{\eps}}}
         f(t-s)\, d\mu(s)  \right|                                                                                                                        \\
  &<  \left| \int_{\widetilde{K_{\eps}}} f(t-s)\,
     d\mu(s) \right| + \|f\|_{\infty} \cdot \frac{\eps}{1+\|f\|_{\infty}}                                                    \\
  &= 0+\eps = \eps
\end{split}
\end{displaymath}
for all $t\notin K_{\eps}:= \widetilde{K_{\eps}}+K_f$.
\end{proof}

An immediate consequence of this is the following result.

\begin{coro} \label{cor:comp}
  \begin{enumerate}
    \item[(i)] Every compactly supported measure is vanishing at infinity.
    \item[(ii)] Let $G$ be compact. Then,  $\cM(G)=\cM_0^{\infty}(G)$.
  \end{enumerate}
\end{coro}

The second part of Corollary \ref{cor:comp} gives rise to the following result.

\begin{prop}
Let $G$ be compact. Then, the space $(\cM_0^{\infty}(G),\|.\|_G)$
is a Banach space, where $\|\mu\|_G := |\mu|(G)$.
\end{prop}
\begin{proof}
It is a known fact that $(\cM(G),\|.\|_G)$ is a Banach space if $G$ is compact. Now,  the claim follows from Corollary \ref{cor:comp} (ii).
\end{proof}

The next question is wether we can obtain similar results for arbitrary locally compact Abelian groups $G$. Let us start by reviewing the product topology for measures \cite[Sec.~2]{ARMA}.

Given a translation bounded measure $\mu$ and a function $f \in C_c(G)$, we have $f* \mu \in C_U(G)$. This allows us to embed $\cM^\infty(G) \hookrightarrow [C_U(G)]^{C_c(G)}$ via
\[
\mu \to \left( \mu*g \right)_{g \in C_c(G)} \,.
\]
Since $(C_U(G), \| .\|_\infty)$ is a Banach space, we can equip $[C_U(G)]^{C_c(G)}$ with the product topology. The topology induced on $\cM^{\infty}(G)$ via the embedding $\cM^\infty(G) \hookrightarrow [C_U(G)]^{C_c(G)}$ is called the \emph{product topology} for measures.

As observed in \cite{ARMA}, the product topology is a locally convex topology given by the family of semi-norms $\{ p_g \}_{g \in C_C(G)}$, where
\[
p_g(\mu):= \| \mu*g \|_\infty \,.
\]

Let us emphasise here that while $[C_U(G)]^{C_c(G)}$ is a complete locally convex topological vector space, the authors of \cite{ARMA} do not know wether $\cM^\infty(G)$ is complete in $[C_U(G)]^{C_c(G)}$. They could only prove that $\cM^\infty(G)$ is bounded closed \cite[Thm.~2.4]{ARMA} and quasi-complete \cite[Cor.~2.1]{ARMA}.

We show that $\cM^\infty_0(G)$ is complete in $\cM^\infty(G)$ and therefore, also bounded closed in $[C_U(G)]^{C_c(G)}$ and quasi-complete.

\begin{theorem} $\cM^\infty_0(G)$ is complete in $\cM^\infty(G)$ with respect to the product topology.
\end{theorem}
\begin{proof}
Let $(\mu_\alpha)_{\alpha}$ be a Cauchy net in $\cM^\infty_0(G)$, which converges to some $\mu \in \cM^\infty(G)$. Then, for each $g \in C_c(G)$ we have
\begin{equation} \label{eq:prod lim}
\mu_\alpha *g \to \mu*g \quad \mbox{ in } \, (C_U(G), \| . \|_\infty) \,.
\end{equation}
Since $\mu_\alpha *g \in C_0(G)$ and $(C_0(G), \| . \|_\infty)$ is complete in $(C_U(G), \| . \|_\infty)$, it follows that $\mu*g \in C_0(G)$. Since $g \in C_c(G)$ is arbitrary, this proves the claim.
\end{proof}

\begin{remark} If $(\mu_\alpha)_{\alpha}$ is an arbitrary Cauchy net in $\cM^\infty_0(G)$, then it is straight-forward to prove that there exists a linear function $L: C_c(G) \to \C$ such that, for all $g \in C_c(G)$, we have
\[
\mu_\alpha *g \to L*g \quad \mbox{ in } \, (C_U(G), \| . \|_\infty) \,.
\]
Moreover, as $C_0(G)$ is closed in $C_U(G)$, it follows that $L*g \in C_0(G)$ for all $g \in C_c(G)$.

The boundedness of the net $(\mu_\alpha)_{\alpha}$ is necessary in order to prove that $L$ is continuous and hence a measure. Indeed, for this we need to show that, for each compact set $K \subseteq G$, there exists some constant $C_K$ such that, for all $f \in C_c(G)$ with $\operatorname{supp}(f) \subseteq K$, we have
\[
\left| L(f) \right| \leq C_K \| f \|_\infty \,.
\]
Note that, since $\mu_\alpha$ is a measure, there exists a constant $C_K(\alpha)$ such that
\[
\left| \mu(f) \right| \leq C_K(\alpha) \| f \|_\infty \,.
\]
The boundedness of the net $(\mu_\alpha)_{\alpha}$ is important for this proof as it means that, for each $K$, the set $\{ C_K(\alpha) \}$ is bounded and hence we can pick the supremum of this set as $C_K$, which yields the continuity of $L$.
\exend
\end{remark}

\medskip

Let us note now that, for a function $f \in C_U(G)$, vanishing at infinity could mean vanishing at infinity as a measure or as a function. We will see that the two concepts are equivalent.

\begin{prop} \label{prop:a}
\begin{displaymath}
\mc M_0^{\infty}(G) \cap C_U(G) = C_0(G) \,.
\end{displaymath}
\end{prop}
\begin{proof}
(i) Obviously, $C_0(G)\subseteq C_U(G)$. Furthermore, $C_0(G)
\subseteq \mc M_0^{\infty}(G)$ because $f\in C_0(G)$ and $g\in
C_c(G)$ imply $f*g\in C_0(G)$.

(ii) Let $f\in\mc M_0^{\infty}(G) \cap C_U(G)$. As before, let $(g_{\alpha})_{\alpha}$ be an approximate identity consisting
of functions from $C_c(G)$. Since $f\in\mc M_0^{\infty}(G)$, we have
\begin{equation} \label{eq:d}
f*g_{\alpha} \in C_0(G)\,.
\end{equation}
Now, as $f\in C_U(G)$ we have
\begin{equation} \label{eq:e}
\lim_{\alpha}\ f*g_{\alpha}\ = \ f \quad \text{ uniformly}.
\end{equation}
Hence, by~(\ref{eq:d}) and~(\ref{eq:e}), we obtain $f\in C_0(G)$.
\end{proof}

\medskip
Let us emphasise that the uniform continuity of $f$ is crucial in the proof of Proposition \ref{prop:a}. For this purpose, we provide an example of a bounded and continuous function which is \emph{not} uniformly continuous, which is vanishing at infinity as a measure but \emph{not} as a function.

\begin{example}
Consider the sequence of functions $(f_n)_{n\in\N}$ defined by
\begin{displaymath}
f_n:\R\to \R,\quad x \mapsto
\begin{cases}
1-2^n\cdot |x-n|, & |x-n|<2^{-n}, \\
0, & \text{otherwise},
\end{cases}
\end{displaymath}
and the function $f(x):=\sum_{n=1}^{\infty} f_n(x)$. A short calculation
shows $\int_{\R} f_n(x)\, dx = \left(\frac{1}{2}\right)^n$, hence
\begin{displaymath}
\int_{\R} f(x)\, dx = \sum_{n=1}^{\infty} \int_{\R} f_n(x)\, dx =1.
\end{displaymath}
Since $f\in L^1(\R)$, the measure $\mu:=f\, \lambda$ is finite and, by Proposition
~\ref{prop:b}, $\mu\in \mc M_0^{\infty}(\R)$. Due to Proposition~\ref{prop:a},
$f$ cannot be uniformly continuous because $f\notin C_0(\R)$.
\end{example}

\medskip

We complete this section by giving a simple criterion for measures with uniformly discrete support to be vanishing at infinity.

\begin{prop}\label{prop:c}
Let $\mu\in \mc M^{\infty}(G)$ and $\Lambda:=\operatorname{supp}
(\mu)$ be uniformly discrete. Then, $\mu\in \mc M_0^{\infty}(G)$
if and only if, for all $\eps>0$, there is a compact set $K_{\eps}\subseteq G$ such that
\begin{displaymath}
|\mu(\{x\})| < \eps \quad \text{ for all } x\in \Lambda\setminus K_{\eps}.
\end{displaymath}
\end{prop}
\begin{proof}
(i) necessary part. Due to the uniform discreteness, there is an open
set $U$ such that
\begin{displaymath}
(U+x)\cap (U+y) = \varnothing \quad \text{ for all distinct } x,y\in \Lambda.
\end{displaymath}
Let $\eps>0$. $\mu\in\mc M_0^{\infty}(G)$ implies that there is a compact
set $K_{\eps}\subseteq G$ such that
\begin{displaymath}
|\mu*f(t)|<\eps \quad \text{ for all } t\notin K_{\eps}.
\end{displaymath}
This is true for all $f\in C_c(G)$. Now, let $f\in C_c(G)$ with $f(0)=1$ and
$\operatorname{supp}(f) \subseteq U$. We obtain
\begin{displaymath}
\begin{split}
|\mu(\{s\})|
  &=\bigg| \sum_{x\in \Lambda\cap (-U+s)} \mu(\{x\})\bigg| = \bigg| \sum_{x\in
     \Lambda\cap (-U+s)} \mu(\{x\})f(s-x) \bigg|                                                                         \\
  &= \bigg| \sum_{x\in\Lambda} \mu(\{x\})f(s-x) \bigg| = |\mu*f(s)| <\eps
\end{split}
\end{displaymath}
for all $s\in \Lambda\setminus K_{\eps}$.

(ii) sufficient part. Let $f\in C_c(G)$ with $K_f:=\operatorname{supp}(f)$.
Let $\eps>0$. By assumption, there is a compact set $K_{\eps}$ such
that
\begin{displaymath}
|\mu(\{x\})|<\eps \cdot \left( \|f\|_{\infty}\cdot \sup_{t\in G} |\Lambda\cap
(-K_f+t)|\right)^{-1}.
\end{displaymath}
For $t\notin \widetilde{K_{\eps}}:=K_f+K_{\eps}$, we obtain
\begin{displaymath}
\begin{split}
|\mu*f(t)|
  &= \big| \sum_{s\in\Lambda} \mu(\{s\}) f(t-s) \big|                                                               \\
  &\leq \sum_{s\in \Lambda\cap (-K_f+t)\cap K_{\eps}} |\mu(\{s\})| |f(t-s)|
     +\sum_{s\in \Lambda\cap (-K_f+t)\cap (\Lambda\setminus K_{\eps})}
      |\mu(\{s\})| |f(t-s)|                                                                                                                \\
  &\leq 0 + \eps \cdot \left( \|f\|_{\infty}\cdot \sup_{t\in G} |\Lambda\cap
       (-K_f+t)|\right)^{-1} \sum_{s\in(-K_f+t)\cap \Lambda} |f(t-s)|     < \eps.
\end{split}
\end{displaymath}
\end{proof}

Proposition \ref{prop:c} suggests the following definition:

\begin{definition} A pure point measure $\mu=\sum_{x \in \Lambda} \mu(\{ x \}) \delta_x$ is said to have \emph{coefficients vanishing at infinity}
if, for all $\eps>0$, there is a compact set $K_{\eps}$ such that
\begin{displaymath}
|\mu(\{x\})| < \eps \quad \text{ for all } x\in \Lambda\setminus K_{\eps}.
\end{displaymath}
\end{definition}

\medskip

\begin{remark} By Proposition \ref{prop:c}, for measures with uniformly discrete support, vanishing at infinity is equivalent to coefficients vanishing at infinity.

For pure point measures without uniformly discrete support, there is no relation in general between these two concepts. Indeed, Example \ref{ex:a} gives a measure $\mu$ which is vanishing at infinity, but for which the coefficients are not vanishing at infinity. On the other hand, we will provide a measure for which the coefficients are vanishing at infinity, but the measure is not, in the next example.
\exend
\end{remark}

\begin{example} Let
\[
\nu:= \sum_{n \in \mathbb N} \sum_{k=0}^{n-1} \frac{1}{n} \delta_{n+\frac{k}{n}} \,.
\]
It is obvious that this measure has coefficients vanishing at infinity.

However, if $f \in C_c(\R)$ is supported inside $(0,1)$, it is easy to see that, by the convergence of Riemann sums, we have
\[
\lim_{x \to \infty} (\nu*f)(x) = \int_0^1 f(t) dt \,.
\]
Therefore, if we pick an $f$ with $\int_0^1 f(t) dt \neq 0$, we get that $\nu*f \notin C_0(\R)$. This shows that $\nu$ is not vanishing at infinity.
\end{example}

\section{The Riemann--Lebesgue lemma for measures}\label{sect:On the}

In this section, we prove a version of the Riemann--Lebesgue lemma for a Fourier transformable measure $\mu$. We then review a similar result of \cite{BF} which complements our result.

The version we prove is important for diffraction theory. While this result does not seem to have been proved explicitly before, the idea of the proof combined with  Proposition \ref{prop:c} has been exploited in some places to prove the existence of singular continuous spectrum \cite{TAO,BG,BGa,GB}.

\begin{theorem}{\rm [Riemann--Lebesgue lemma for measures]}  \label{thm:a}
Let $\mu\in \mc M(G)$ be a Fourier transformable measure. If $\widehat{\mu}$
is absolutely continuous, then $\mu\in\mc M_0^{\infty}(G)$.
\end{theorem}
\begin{proof}
Let $g\in K_2(G)$. By \cite[Thm. 3.1]{ARMA1}, $\mu*g$ is Fourier transformable
and
\begin{displaymath}
\widehat{\mu*g} = \widehat{\mu}\, \widehat g.
\end{displaymath}
The measure $\widehat{\mu}$ is absolutely continuous, i.e.
\begin{displaymath}
\widehat{\mu} = f\, d\theta\quad \text{ with } f\in L_{\text{loc}}^1(G)
\end{displaymath}
due to the Radon--Nikodym theorem. Furthermore, by the definition of the transformability of $\mu$, we have
\begin{displaymath}
\int_G\check g f\, d\theta = \int_G \check g\, d\widehat{\mu}<\infty \,.
\end{displaymath}
Hence, $\widehat{\mu}\widehat g$ is a finite measure and therefore $f\widehat g \in L^1(G)$. In particular,  $f\widehat g$ is Fourier transformable as $L^1$-function. This implies \cite{ARMA1,MoSt}
\begin{displaymath}
(\mu*g)^{\dagger} = \widehat{f\widehat g}\,,
\end{displaymath}
where for a function $h$ we denote by $h^\dagger$ the reflection
\[
h^\dagger(x):=h(-x) \,.
\]

Now, by the Riemann--Lebesgue lemma for functions, we have $(\mu*g)^{\dagger} \in C_0(G)$. The claim follows by
an application of Lemma~\ref{lem:a}.
\end{proof}

\medskip

\begin{remark}\label{rem:a}
If $\mu\in\mc M_0^{\infty}(G)$, its Fourier transform is not necessarily absolutely
continuous. Consider the measure
\begin{displaymath}
\mu:=2\pi J_0(2\pi\|.\|)\, \lambda^2,
\end{displaymath}
where $J_0$ is the Bessel function of the first kind of order zero. By
\cite[p. 154]{SW}, we obtain $\widehat{\lambda^2|_{S^1}}=\mu$ and,
consequently,
\begin{displaymath}
\widehat{\mu} = \widetilde{\lambda^2|_{S^1}}.
\end{displaymath}
As $2\pi J_0(2\pi\|.\|)\in C_0(\R^2)$, Proposition~\ref{prop:a} implies
$2\pi J_0(2\pi\|.\|)\, \lambda^2\in\mc M_0^{\infty}(\R^2)$. But
$\widehat{\mu}$ is singular continuous.
\exend
\end{remark}

\medskip

Let us start by looking at some immediate consequences of Theorem \ref{thm:a}.

First, by combining Theorem \ref{thm:a} with Proposition \ref{prop:c} we get

\begin{coro} Let
\[
\mu:= \sum_{x \in \Lambda} \mu(\{ x \}) \delta_x \,,
\]
be a Fourier transformable measure. If $\Lambda$ is uniformly discrete and $\widehat{\mu}$ is absolutely continuous, then the coefficients of $\mu$ are vanishing at infinity.  \qed
\end{coro}

\medskip

\begin{coro}\label{cor:b} Let
\[
\mu:= \sum_{x \in \Lambda} \mu(\{ x \}) \delta_x \,,
\]
be a Fourier transformable measure with lattice support. If $\widehat{\mu}$ is absolutely continuous, then the coefficients of $\mu$ are vanishing at infinity.    \qed
\end{coro}

\medskip

\begin{coro} Let
\[
\mu:= \sum_{x \in \Lambda} \mu(\{ x \}) \delta_x \,,
\]
be a Fourier transformable measure with Meyer set support. If $\widehat{\mu}_{\operatorname{sc}}=0$, then the coefficients of $\mu_{0}$ are vanishing at infinity.
\end{coro}
\begin{proof}
Since $\mu$ has Meyer set support, so has $\mu_{0}$ \cite{NS5}. Therefore, $\mu_{0}$ is a measure vanishing at infinity with uniformly discrete support. The claim now follows from Corollary \ref{cor:b}.
\end{proof}

\medskip

We look next at a Fourier dual version of Theorem \ref{thm:a}.

\begin{coro}\label{cor:a} Let $\mu$ be a measure which is twice Fourier transformable. If $\mu$ is absolutely continuous, then $\widehat{\mu} \in \cM^\infty_0(\widehat{G})$.   \qed
\end{coro}

We complete the section by recalling a result of \cite{BF}. This result shows that in Corollary~\ref{cor:a}, twice Fourier transformability  can be replaced by positive definitedness, and can also be seen as the Fourier dual of a particular case in Theorem \ref{thm:a}.

\begin{theorem}\label{thm:b}\cite[Prop.~4.9]{BF} Let $\mu$ be a positive definite measure. If $\mu$ is absolutely continuous, then $\widehat{\mu} \in \cM^\infty_0(\widehat{G})$.   \qed
\end{theorem}

\section{Few more examples}\label{sect:exa}

In this section, we will look at few more examples and provide a simple method of creating measures vanishing at infinity.

We start by reviewing an example from \cite{BF}.

\begin{example}\label{ex:BF}\cite[Remark 1.15]{BF} For each positive integer $n$ and each integer $0 \leq k \leq 2^n-1$, define $f(x)=(-1)^k$ for all $x \in [n+\frac{k}{2^n}, n+\frac{k+1}{2^n})$. Next, define $f(x)=0$ at all other $x \in \R$.

Then $\mu=f \lambda \in \cM^\infty_0(\R)$.
\end{example}

\begin{example}\label{ex:b} Let $\mu_n=\frac{1}{n} \sum_{k=1}^n \delta_{n+\frac{k}{n}}$. Define
\begin{displaymath}
\mu=\lambda-\sum_{n=1}^\infty (\mu_n+\widetilde{\mu_n})
\end{displaymath}
Then $\mu \in \cM^\infty_0(\R)$. This is true for the following reason.

By the convergence of the Riemann Sums for the integrals of continuous function, in the vague topology we have
\begin{displaymath}
\lim_{n\to\infty} \left( \frac{1}{n} \sum_{k=1}^n \delta_{\frac{k}{n}} \right) \,=\, \lambda_{[0,1]} \,.
\end{displaymath}
This shows that $\lim_{x \to \infty} \mu*f(x)=0$, which proves the claim.
\end{example}

\begin{example}
Consider the sequence of measures $(\mu_n)_{n\in\N}$ defined by
\begin{displaymath}
\mu_n:=\delta_0+\lambda|_{[n,n+1]}.
\end{displaymath}
Applying \cite[Thm. 1.3.3(b)]{Rud}, we obtain
\begin{displaymath}
\begin{split}
\reallywidehat{(\mu_n*\widetilde{\mu_n})}
  &=\widehat{\mu_n}\cdot \widehat{\widetilde{\mu_n}} = |\widehat{\mu_n}|^2
     = \big| \widehat{\delta_0} + \widehat{\lambda|_{[n,n+1]}}\big|^2                            \\
  &= \big| 1+\e^{-\pi\im x(2n+1)}\operatorname{sinc}(\pi x)\big|^2 \, \lambda              \\
  &= \left( 1+2\cos\big(\pi x(2n+1)\big)\operatorname{sinc}(\pi x) +
        \operatorname{sinc}^2(\pi x)\right)\, \lambda.
\end{split}
\end{displaymath}
Now, if we define
\begin{displaymath}
\mu:=\sum_{n=0}^{\infty}\frac{1}{2^n}\, \mu_n*\widetilde{\mu_n},
\end{displaymath}
we get
\begin{displaymath}
\widehat{\mu}=\sum_{n=0}^{\infty} \frac{1}{2^n} \left( 1+2\cos
\big(\pi x(2n+1)\big)\operatorname{sinc}(\pi x) +
        \operatorname{sinc}^2(\pi x)\right)\, \lambda
\end{displaymath}
Thus, $\widehat{\mu}$ is absolutely continuous and, by Theorem
~\ref{thm:a}, $\mu\in\mc M_0^{\infty}(\R)$.
\end{example}

We complete the section by introducing a large class of measures vanishing at infinity, which contains the measures from Example \ref{ex:a}, Example \ref{ex:BF} and Example \ref{ex:b}.

\begin{prop}\label{prop: generating measures at infty} Let $(\mu_n)_{n\in\N}$ a sequence of measures, $K \subseteq G$ be a compact set, and $(t_n)_{n\in\N}$ a sequence in $G$ with the following properties:
\begin{itemize}
  \item [(i)] $\operatorname{supp}(\mu_n) \subseteq K$ for all $n$.
  \item [(ii)] $\left| \mu_n \right|(K)$ is bounded.
  \item [(iii)] In the vague topology, we have
  \begin{displaymath}
\lim_{n\to\infty} \mu_n =0 \,.
  \end{displaymath}
  \item [(iv)] The $t_n$ are distinct and the set $\{ t_n | n \in \N \}$ is uniformly discrete.
\end{itemize}
Then, the measure
\begin{displaymath}
  \mu:= \sum_{n=1}^{\infty} T_{t_n} \mu_n \,,
\end{displaymath}
is vanishing at infinity, where, for a measure $\mu$, the translation $T_t \mu$ is defined via
\[
[T_t(\mu)](f):= \mu(T_{-t}f) \,.
\]
\end{prop}
\begin{proof}
First, it follows immediately from the assumptions that
\begin{displaymath}
\| \mu \|_K < \infty
\end{displaymath}
and hence $\mu$ is a translation bounded measure.

Fix some arbitrary $f \in C_c(G)$ and let $\eps >0$. Let $K_0$ be a compact set such that $\operatorname{supp}(f) \subseteq K_0=-K_0$.  By (iv), there exists some $M\ge0$ such that for all $x \in G$ there are at most $M$ elements in
\begin{displaymath}
\{ t_n\, |\, n \in \N \} \cap (x+K_0-K) \,.
\end{displaymath}
Since $\mu$ is translation bounded, $\mu*f$ is uniformly continuous. Therefore, there exists some open neighbourhood $U$ of $0$ such that, for all $t,s \in G$ with $t-s \in U$, we have
\begin{displaymath}
\left| \mu*f(t)-\mu*f(s) \right| < \frac{\eps}{2M} \,.
\end{displaymath}
Next, by the compactness of $K+K_0$, we  can find a finite set $x_1,...,x_m$ such that
\begin{displaymath}
K+K_0 \subseteq \bigcup_{j=1}^m x_j+U \,.
\end{displaymath}

Finally, by the vague convergence of $\mu_n$, there exists some $N_\eps$ such that, for all $n >N_\eps$ and all $1 \leq j \leq m$, we have
\begin{displaymath}
\left| \mu_n(T_{x_j} f^\dagger) \right| < \frac{\eps}{2M} \,.
\end{displaymath}

Let $K_\eps := \bigcup_{n=1}^{N_\eps} t_n+(K+K_0)$. Then, $K_\eps$ is a compact set. We show that $\left| \mu*f \right| < \eps$ outside $K_\eps$. Let $x \notin K_\eps$. Then
\begin{equation}\label{EQ1}
\left| \mu*f(x) \right| \leq \sum_{n} \left| T_{t_n} \mu_n*f(x) \right|  = \sum_{n \in F} \left| T_{t_n} \mu_n*f(x) \right|
\end{equation}
where $F:= \{ n\, |\, T_{t_n} \mu_n*f(x) \neq 0 \}$. Since $\{ t_n\, |\, n \in F\} \subseteq \{ t_n \} \cap (K_0-K)$, we know that $|F| \leq M$.

Moreover, for each $n \in F$, we have $n > N_\eps$, therefore
\begin{displaymath}
\left| \mu_n*f(x_j) \right|=\left| \mu_n(T_{x_j} f^\dagger) \right| < \frac{\eps}{2M} \,.
\end{displaymath}
Finally, if $n \in F$, we have $x -t_n \in K+K_0 =  \bigcup_{j=1}^m x_j+U$. Therefore, there exists some $j$ and $u \in U$ such that
\begin{displaymath}
x-t_n =x_j+u \,.
\end{displaymath}
Hence,
\begin{displaymath}
\left| T_{t_n} \mu_n*f(x) \right|  = \left| \mu*f (x_j+u) \right| \leq \left| \mu*f (x_j) \right| + \frac{\eps}{2M} < \frac{\eps}{2M } + \frac{\eps}{2M } =\frac{\eps}{M}.
\end{displaymath}
Therefore, by (\ref{EQ1}), we get
\begin{displaymath}
\left| \mu*f(x) \right| < \frac{\eps}{M} |F| \leq \eps \,.
\end{displaymath}
This completes the proof.
\end{proof}

\begin{remark}
\begin{enumerate}
\item[(i)] Setting $t_{n}=n$, $K=[0,1]$ and
\[
\mu_n =\sum_{k=0}^{2^n-1} (-1)^k \lambda_{[n+\frac{k}{2^n}, n+\frac{k+1}{2^n})} \,,
\]
in Proposition \ref{prop: generating measures at infty}, we get exactly the Example \ref{ex:BF}.
\item[(ii)] Setting $t_{2n}=n, t_{2n+1}=-n$, $K=[-1,1]$ and
\begin{displaymath}
\mu_{2n}=\delta_{\frac{1}{n}}-\delta_0 \quad ; \quad \mu_{2n+1}=\delta_{-\frac{1}{n}}-\delta_0 \,,
\end{displaymath}
in Proposition \ref{prop: generating measures at infty}, we get exactly the Example \ref{ex:a}.
\item[(iii)] Let us denote by $\lambda_{[0,1)}$ the restriction of the Lebesgue measure to $[0,1)$. Setting exactly as in (i) $t_{2n}=n,\, t_{2n+1}=-n$, $K=[-1,1]$ and
\begin{displaymath}
\mu_{2n}=\lambda_{[0,1)}-\left(\frac{1}{n}\sum_{k=1}^n \delta_{\frac{k}{n}}\right)  \quad ; \quad \mu_{2n+1}=\widetilde{\mu_{2n}} \,,
\end{displaymath}
for $n\ge0$ in Proposition \ref{prop: generating measures at infty}, we get exactly the Example \ref{ex:b}.\exend
\end{enumerate}
\end{remark}

\section{Connection to Rajchman Measures}\label{sect:raj}

In this section, we will see that there is a strong connection between measures vanishing at infinity and Rajchman measures. Therefore, let us recall the definition of a Rajchman measure \cite{RuLy}.

\begin{definition}
A finite measure $\mu$ is called a \emph{Rajchman measure} if $\widehat{\mu} \in C_0(\widehat{G})$.

We denote the class of Rajchman measures by $\cR(G)$.
\end{definition}

\begin{remark} It is well-known that every finite measure which is absolutely continuous is a Rajchman measure. The converse is not true. There are Rajchman measures that are not absolutely continuous, see \cite{RuLy} and references therein.
\exend
\end{remark}

A stronger version of Theorem \ref{thm:a} is given by the following theorem, which shows the connection between measures vanishing at infinity and Rajchman measures.

\begin{theorem}\label{thm:c} Let $\mu$ be a Fourier transformable measure. Then, $\mu \in \cM_0^\infty(G)$ if and only if for all $f \in C_c(G)$ we have
\begin{displaymath}
\big| \widehat{f}\, \big|^2 \widehat{\mu} \in \cR(\widehat{G}) \,.
\end{displaymath}
\end{theorem}
\begin{proof}
Let $f \in C_c(G)$. Then, since $\mu$ is Fourier transformable, $\mu*f*\widetilde{f}$ is the inverse Fourier transform of the finite measure $\big| \widehat{f}\, \big|^2 \widehat{\mu}$. Therefore, by the definition of $\cR(\widehat{G})$ we get that $\mu*f*\widetilde{f} \in C_0(G)$ if and only if $\big| \widehat{f}\, \big|^2 \widehat{\mu} \in \cR(\widehat{G})$. The claim follows now from Lemma \ref{lem:a}.
\end{proof}

\section*{Acknowledgments}

We would like to thank Michael Baake for pointing out the connection to Rajchman measures, and for some suggestions which improved the quality of the manuscript. The first author was supported by the German Research Foundation (DFG) via the Collaborative Research Centre (CRC 701) through the faculty of Mathematics, Bielefeld University. The second author was supported by NSERC with the grant number 2014-03762, and the authors are grateful for the support. The project was started when TS visited Department of Mathematics, Trent University, and was completed when NS visited Department of Mathematics, Bielefeld University, and the authors would like to thank the departments for hospitality.

\end{document}